\documentclass[reprint,floatfix,longbibliography,twocolumn,notitlepage]{revtex4-2}
\usepackage{graphicx} %
\usepackage[T1]{fontenc}
\usepackage[colorlinks=true,linkcolor=red, citecolor=blue]{hyperref}
\usepackage{amssymb}
\usepackage{braket}
\usepackage{physics}
\usepackage{subcaption}
\usepackage{xcolor}
\usepackage{mathtools}
\usepackage{amsmath}
\usepackage[scr=rsfs,cal=boondox]{mathalfa}

\usepackage{multirow}
\usepackage{booktabs}

\usepackage{cleveref}

\newcommand{\Uh}{\mathbb{U}/H}
\newcommand{\Haar}{\mathrm{Haar}}
\usepackage{amsthm}
\usepackage{tikz}
\usepackage{float}
\usetikzlibrary{quantikz}

\newtheorem{theorem}{Theorem}[section]
\newtheorem{remark}[theorem]{Remark}
\newtheorem{lemma}[theorem]{Lemma}
\newtheorem{definition}[theorem]{Definition}

\newtheorem{corollary}[theorem]{Corollary}

\begin{document}

\title{A quantum tug of war between randomness and symmetries on homogeneous spaces}

\author{Rahul Arvind}
\affiliation{%
 Institute of High Performance Computing (IHPC), Agency for Science, Technology and Research (A*STAR), 1 Fusionopolis Way, \#16-16 Connexis, Singapore 138632, Republic of Singapore
}%
\author{Kishor Bharti}
\affiliation{%
 Institute of High Performance Computing (IHPC), Agency for Science, Technology and Research (A*STAR), 1 Fusionopolis Way, \#16-16 Connexis, Singapore 138632, Republic of Singapore
}
\author{Jun Yong Khoo}
\affiliation{%
 Institute of High Performance Computing (IHPC), Agency for Science, Technology and Research (A*STAR), 1 Fusionopolis Way, \#16-16 Connexis, Singapore 138632, Republic of Singapore
}
\author{Dax Enshan Koh}
\affiliation{%
 Institute of High Performance Computing (IHPC), Agency for Science, Technology and Research (A*STAR), 1 Fusionopolis Way, \#16-16 Connexis, Singapore 138632, Republic of Singapore
}
\author{Jian Feng Kong}
\affiliation{%
 Institute of High Performance Computing (IHPC), Agency for Science, Technology and Research (A*STAR), 1 Fusionopolis Way, \#16-16 Connexis, Singapore 138632, Republic of Singapore
}

\begin{abstract}
\normalsize

We explore the interplay between symmetry and randomness in 
quantum information. Adopting a geometric approach, we consider states as $H$-equivalent if related by a symmetry transformation characterized by the group $H$. We then introduce the Haar measure on the homogeneous space $\mathbb{U}/H$, characterizing true randomness for $H$-equivalent systems. While this mathematical machinery is well-studied by mathematicians, it has seen limited application in quantum information: we believe our work to be the first instance of utilizing homogeneous spaces to characterize symmetry in quantum information. This is followed by a discussion of approximations of true randomness, commencing with $t$-wise independent approximations and defining $t$-designs on $\mathbb{U}/H$ and $H$-equivalent states. Transitioning further, we explore pseudorandomness, defining pseudorandom unitaries and states within homogeneous spaces. Finally, as a practical demonstration of our findings, we study the expressibility of quantum machine learning ansatze in homogeneous spaces. Our work provides a fresh perspective on the relationship between randomness and symmetry in the quantum world.
\end{abstract}

\maketitle

\section{Introduction}

Symmetry is a powerful tool for approaching various problems in physics, often allowing us to reduce complicated systems into much simpler ones \cite{gross1996role}. For example, symmetries in quantum field theory enable us to derive conserved charges, while in high-energy physics, symmetry breaking is used to explain the Higgs mechanism \cite{Englert_1964, Pich_2016}. In quantum physics, we often use symmetries when solving for the dynamics of systems, where a standard trick is to use the result that commuting observables share the same eigenvectors. Its applications in quantum information and computing have been vast, such as in randomized benchmarking \cite{onorati2019randomized} and quantum machine learning \cite{schatzki2022theoretical, ragone2023_representation, kazi2023_universality, nguyen2022_theory,PRXQuantum.3.030341, west2022_reflection, mernyei2022_equivariant, Zheng_2023, sauvage2022_building}.  At the same time, symmetry tends to decrease the randomness in a system. Consider for instance a simple one-dimensional Ising chain with $N$ sites. If we impose the symmetry constraint that the even sites all have the same eigenvalue $P_{E}$ for the parity operator (and similarly $P_{O}$ for the odd sites), the set of microstates that this system can be in is simply $\{ -1, 1\}^{2}$, a much smaller set than the $2^{N}$-many that we would have had if there was no symmetry constraint. As a result the entropy of the system diminishes greatly. In this sense, randomness and symmetry seem to \textit{oppose} each other. Intuitively, if we were to think of randomness in a system as a lack of information about the description of some probabilistic system, one could think of the symmetry as providing some additional information about it. For a pictorial view, one can think of the entire state space as a rope, with randomness and symmetry pulling it in opposite directions (Figure \ref{fig:symmetry}).

There are multiple ways to think about symmetry in a quantum system. One such way is commutation ---  we say that $H \leq \mathbb{U}(d)$ is a symmetry group if all its elements $v\in H$ satisfy $[v,u] = 0$ for all $u \in \mathbb{U}(d)$, the unitary group on $d$ dimensions, where the symbol $\leq$ indicates the subgroup relation. The notion of symmetric designs in this context has been defined in \cite{mitsuhashi2023clifford}. In our work, we take a more geometric approach. Let $H$ be the symmetry group in question that we want to work with, and let us call two states $H$-equivalent if it is possible to obtain one of the states from the other by left multiplication with an element in $H$. Then (in some applications) the two $H$-equivalent states \textit{might as well} be the same --- whether we have one or the other, we can obtain the other one using a symmetry transformation. This invites us to think of the homogeneous space $\mathbb{U}/H$, or the set of all $H$-equivalent unitaries, where the equivalence is under left-multiplication by the elements in $\mathbb{U}$. This homogeneous space is the set of all left cosets of $\mathbb{U}$ under the subgroup $H$; put more simply, it contains a set of equivalence classes, in which each class contains a set of unitaries which are $H$-equivalent to each other. In the context of quantum computing, what we are doing becomes clearer: here we start often with the $\ket{0}^{\otimes N}$ state, and use our quantum circuit to apply unitary operators to it. If we decide that we do not need our circuit to produce every state in the Hilbert space --- rather, we only want our circuit to produce all the $H$-inequivalent states --- then it suffices for the circuit to just generate one representative of $\mathbb{U}/H$. We shall use this idea later in the paper to discuss the expressibility of a quantum learning task.

In our case, randomness for a quantum system refers to picking states and operators at random --- we would say, for example, that an ensemble of operators is random if the ensemble was not biased towards any particular operator. In a certain sense, possession of a truly uniform ensemble should not provide the user any knowledge whatsoever about what one might obtain when sampling from this ensemble. The Haar measure (and hence the ensemble of unitaries that one gets when sampling according to this measure) provides us with such a truly uniform ensemble, in the sense that it is invariant under unitary multiplication. Indeed, the Haar measure plays an important role in quantum information and quantum computation \cite{mele2023introduction}, formalizing the notion of generating unitary operators uniformly at random. Its significance extends across various applications, including randomized benchmarking \cite{magesan2012characterizing}, shadow tomography \cite{huang2020predicting,koh2022classical,chen2021robust,hu2023classical}, quantum machine learning \cite{mcclean2018barren}, near-term quantum algorithms~\cite{bharti2021noisy}, pseudorandomness \cite{ji2018pseudorandom}, and evidence for quantum advantage \cite{hangleiter2023computational,hangleiter2018anticoncentration,bouland2018complexity}.

We begin our paper with a discussion of integration on homogeneous spaces of locally compact groups, defining a Haar measure on these spaces. This is the main mathematical machinery that we will use throughout the paper. It has been well studied by mathematicians, but its use in quantum information has been limited so far. To our knowledge, this is the first time homogeneous spaces have been used to characterize symmetry in quantum information. This establishes \textit{true randomness} on these homogeneous spaces --- if we are to sample from $\mathbb{U}/H$ according to the defined Haar measure, we will get a sample that is truly random in the sense of $\mathbb{U}/H$. We then proceed to approximations of true randomness, beginning with $t$-wise independent approximations of randomness, defining $t$-designs on $\mathbb{U}/H$ and $H$-equivalent states. We then switch to pseudorandomness, defining homogeneous space pseudorandom unitaries. We then provide an application of this apparatus to quantum machine learning, and finish off with a conclusion summarizing the work and providing an outlook. 

\begin{widetext}
\begin{minipage}{\linewidth}
\begin{table}[H]
    \centering
    
    \begin{tabular}{llll}
        \toprule
        & Classical & Quantum &  Quantum with symmetry \\
        \midrule
        True randomness & Uniform distribution & Haar measure & Haar measure on homogeneous spaces \\
        t-wise independence & t-wise independent  random variables & Quantum t-designs & Homogeneous space t-designs (\textit{this work}) \\
        
        Pseudorandomness & PRGs & PRS’s~\cite{ji2018pseudorandom}  & HPRS's (\textit{this work})\\
        & PRFs, PRPs & PRUs~\cite{ji2018pseudorandom}  & HPRUs (\textit{this work})\\
        \bottomrule
    \end{tabular}
    \caption{A quantum tug of war between randomness and symmetries while approximating true randomness}
\end{table}

\begin{figure}[H]
\centering
\includegraphics[width=0.75\textwidth]{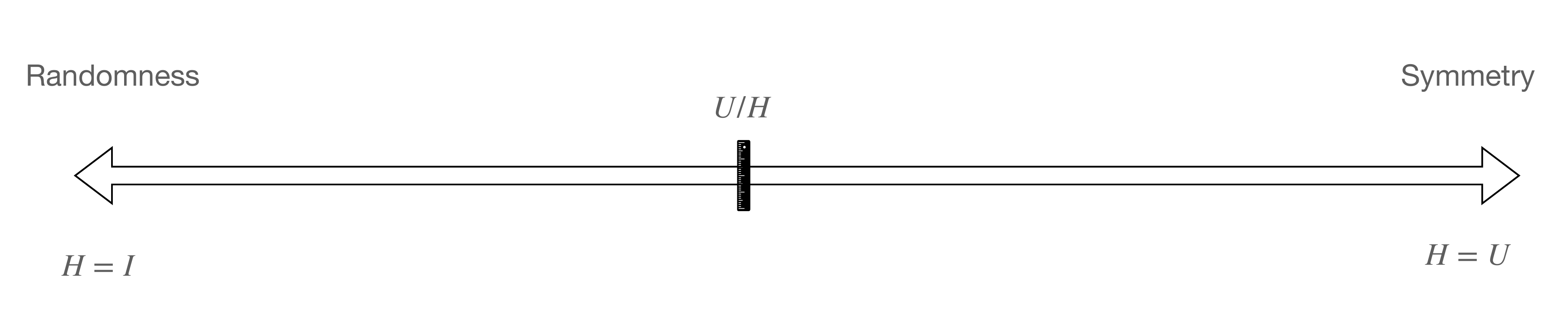}
\caption{Pictorial representation of a quantum tug of war between randomness and symmetries on homogeneous spaces.}
\label{fig:symmetry}
 \end{figure}
\end{minipage}
\end{widetext}

\section{Preliminaries}

\subsection{Haar measure on homogeneous spaces}

Let $G$ be a locally compact topological group, and $H \leq G$ be its subgroup. Then there is a (trivial) left action of $H$ on $G$, defined as a map from $H \times G \rightarrow G$ that takes $(h,g) \mapsto hg$. It is easy to check that this satisfies the axioms of a group action: the identity is present in $H$ (as it is a subgroup), and the compatibility axiom just reduces to the associativity in the group.

Now we can define the space $G/H$ as the space of all left cosets of $G$, where two elements $g_{1}, g_{2}$ of $G$ are equivalent (and hence in the same coset) if $\exists$ $h \in H$ such that $g_{1} = hg_{2}$. Note that $H$ itself forms one such coset, because of the closure of group multiplication in $H$. One often writes cosets as $gH$, where $gH$ refers to that coset which contains $g$. This space is called the homogeneous space of $H$ in $G$.

\begin{theorem}[\cite{nachbin_1985}.] Let $G$ be a locally compact Hausdorff topological group and $H \leq G$. Let $\mathrm{d}^{l}(G)$ and $\mathrm{d}^{l}(H)$ be Haar measures on $G$ and $H$ respectively.

\begin{enumerate}
    \item There is a non-zero $G$-invariant Borel measure $\mathrm{d}_{G/H}$ on $G/H$ if and only if the modular functions of $G$ and $H$ satisfy $\delta_{H} = \delta{G}|_{H}$,
    \item This measure is unique and can be normalized such that 
    \begin{equation}
        \int_{G} f(x) \mathrm{d}^{l}_{G}(x) = \int_{G/H} \bar{f}(xH) \mathrm{d}_{G/H}(xH),
    \end{equation}
where $\bar{f}$ can be obtained from $f$ using $\bar{f}(xH) = \int_{H} f(xh) \mathrm{d}^{l}_{H} (h)$.
\end{enumerate}
\end{theorem}

Here we have introduced the modular function, which is a continuous group homomorphism from $G$ to $(\mathbb{R}^{+}, \times)$ defined in \cite{nachbin_1985, Echterhoff_Deitmar_2014}.

In our case, we will have $G$ to be the unitary group and $H$ to be the group of symmetries that we want to eliminate --- we can then construct a measure on this reduced coset space that is $G$ invariant (and hence uniform in the \textit{right} sense) as desired using the above theorem. Note that for this, we need to satisfy condition $1$. This is easy in our case because of the following simple lemma:

\begin{lemma}
 Let $G$ be a compact Hausdorff topological group. Then the modular function $\delta$ of $G$ is identically $1$.
\end{lemma}

\begin{proof}
    Assume not. Then the modular function must take on some value $k \neq 1$ for some $g \in G$. Now if this $k > 1$, then we can use the homomorphism structure of the modular function to write $\delta(g^{n}) = (\delta(g))^{n}$, and hence $\delta$ is arbitrarily large in $G$. Now we use the continuity of $\delta$ --- since continuous functions map compact sets to a compact image, this is a contradiction. If $k \leq 1$, we can convert it to the $k > 1$ case by noting that $\delta(e) = 1$ (group homomorphism maps identity to identity), and so $\delta(g^{-1}g) = 1$. But again using the homomorphism structure, $\delta(g^{-1}g) = \delta(g^{-1})\delta(g) = k\delta(g^{-1})$ so $\delta(g^{-1}) > 1$. 
\end{proof}

Since the unitary group is compact, all we have to do is be careful to choose $\textit{closed}$ subsets for our symmetry group, and these will be compact as well (It can be easily shown that a closed subset of a compact set is compact).

\section{Sampling from Haar Measure on \texorpdfstring{$\mathbb{U}/H$}{U/H}}
Let $G$ be a locally compact topological group (of which $\mathbb{U}$ forms one example), assume $H$ is a closed subgroup, and denote for the Haar measures on $G/H$ and $G$ $\mathrm{d}_{G/H}$ and $\mathrm{d}_{G}$ respectively. Let $f$ be a function from $G/H \rightarrow \mathbb{R}$. Then it can be shown that the following integrals are equal to each other \cite{nachbin_1985, Echterhoff_Deitmar_2014}: 
\begin{equation}
    \int_{G/H} f(xH) \mathrm{d}_{G/H}(xH) = \int_{G} F(x) \mathrm{d}^{l}_{G}(x),
    \label{inversion}
\end{equation}
where
\begin{equation}
    f(xH) = \int_{H} F(xh) \mathrm{d}^{l}_{H} (h).
\end{equation}

It can further be shown that the mapping $F \mapsto f$ is a well-defined surjection, so that we can always construct at least one $F$ for every $f$ \cite{Echterhoff_Deitmar_2014}. Let $\mathrm{Supp}(f) \leq G/H$ be the support of $f$, and let $K = \pi^{-1}(\mathrm{Supp}(f))$, %
where $\pi$ is the canonical map $\pi : G \rightarrow G/H$. Let $\phi$ be a continuous function on $G$ such that $\phi(K) = 1$ and $\phi \geq 0$. Then we can define for $x \in G$,

\begin{equation}
    F(x) = \frac{f(xH)\phi(x)}{\int_{H} \phi(xh) \mathrm{d}_{H}(h)}.
\end{equation}

So then the expression 
\begin{align}
    \mu(S) &= \int_{G/H} 1_{S}(xH) \mathrm{d}_{G/H}(xH)
    \nonumber\\
    &= \int_{G} \frac{1_{\tilde{S}}(x)\phi(x)}{\int_{H} \phi(xh) \mathrm{d}_{H}(h)} \mathrm{d}_{G}(x) = \int_{G} 1_{\tilde{S}}(x) \mathrm{d}_{G}(x),
    \label{bigsubgrouphaarres}
\end{align}
where we have chosen $\phi(G) = 1$, which is clearly continuous, and then normalized the measure. Here the set $\bar{S} = \pi^{-1}(S)$. This equation provides a way to perform Haar sampling on $\mathbb{U}/H$ given that we can Haar sample on $\mathbb{U}$, since it relates the volume that the measure assigns to $S \leq \mathbb{U}/H$ with $\pi^{-1}(S) \leq \mathbb{U}$. This immediately also gives us a physical realization of Haar randomness on $\mathbb{U}/H$ using random quantum circuits and post-processing: if we et $W$ be a random quantum circuit such that $W$ is approximately $\mathrm{Haar}$ on $\mathbb{U}(d)$. Then $W$ followed by post-processing wherein we pick the equivalence class that the particular realization of $W$ belongs to is approximately Haar random on $\mathbb{U}(d)/H$.

\section{\texorpdfstring{$t$}{t}-Designs for \texorpdfstring{$\mathbb{U}/H$}{U/H}}

While the process of generating Haar random states is resource intensive, in numerous scenarios, we only need to access the lower-order moments of the Haar measure. In such situations, $t$-designs have emerged as a useful approach for representing Haar random distributions, accurately replicating the moments of the Haar measure up to the $t$-th order. This motivates us to
define state $t$-designs as well where we take into account symmetry under $H$. However, here we must be careful to make sure that the quantities we are dealing with are well-defined. Let us build this up in parts, starting with state $t$-designs, and then providing a general homogeneous space $t$-design.

\subsection{State \texorpdfstring{$t$}{t}-designs}
We begin with a definition for $H$-equivalence of states.

\begin{definition}
    Two states $\ket{\psi_{1}}, \ket{\psi_{2}}$ are $H$-equivalent if $\exists h \in H$ such that $\ket{\psi_{1}} = h \ket{\psi_{2}}$.
    \label{hequivstate}
\end{definition}

It is easy to check that this is an equivalence relation. Reflexivity is guaranteed by the fact that $H$ is a subgroup and has an identity, the existence of an inverse in $H$ guarantees symmetry, and transitivity follows from the composition of unitary operations. Thus we have a set of equivalence classes for the states as well, and in particular, there exists $\tau$ that maps a pure state to its equivalence class. Now let us define the action of an element of $\mathbb{U}/H$ on a state:

\begin{definition}
    Let $x \in \mathbb{U}/H$, and let $\tau$ be the map of a state to its equivalence class under $H$. Then we define $x\ket{\psi}$ as $\tau(U \ket{\psi})$, where $U$ is some representative of $x$.
    \label{eqcmult}
\end{definition}

We can check that this is independent of the representative as well, because if $U$ and $V$ are representatives of $x$, then $U = hV$ for some $h \in H$, and so $U\ket{\psi} = hV\ket{\psi}$ and so $\tau(U\ket{\psi}) = \tau(hV \ket{\psi})$. But $\tau (hV \ket{\psi}) = \tau (V \ket{\psi})$ since $\tau$ maps $H$-related states to the same equivalence class.

\begin{definition}
    Let $(p_{i}, \ket{\psi_{i}})_{i \in \mathcal{I}}$ be a distribution of states (so that $\sum_{i \in \mathcal{I}} p_{i} = 1$), $H$ be a closed (more generally, unimodular) subset of $\mathbb{U}$. Then $(p_{i}, \ket{\psi_{i}})_{i 
    \in \mathcal{I}}$ is a $H$-homogeneous space state $t$-design if 
    \begin{equation}
        \sum_{i \in \mathcal{I}} p_{i} \Big(\ketbra{\psi_{i}}\Big)^{\otimes t} = \int_{\mathbb{U}/H}  \Big( \phi \bigg(\ketbra{x} \bigg)\Big)^{\otimes t} \dd_{\mathbb{U}/H}(x),
    \end{equation}
where $\dd_{\mathbb{U}/H}$ is the Haar measure on $\mathbb{U}/H$, and $\ket{x}$ is the equivalence class of states (which we have extended to density matrices here) that one gets when computing $x\ket{0}$ as in Definition \ref{eqcmult}. We have also defined $\phi$ to be an arbitrary (possibly probabilistic) map that sends equivalence classes of density matrices to a particular density matrix. \end{definition}

Operationally, we can think of this as an ensemble of states which is indistinguishable from the Haar distribution on $\mathbb{U}/H$ \textit{up to left multiplication by elements of $H$} when only $t$ copies of states are used. We should note however that some of the tricks that we possess with unitary designs over $\mathbb{U}$ do not quite work in the same way. We discuss this explicitly when we try to define a measure of expressibility that is similar to that in \cite{Sim_2019} in Appendix~\ref{sec:limitations}. Similar to the exact case, we can easily extend the definition of a $t$-design to the approximate case:
\begin{definition}
    Let $\{(p_{i}, \ket{\psi_{i}})\}_{i \in \mathcal{I}}$ be a distribution of states (so that $\sum_{i \in \mathcal{I}} p_{i} = 1$), $H$ be a closed (more generally, unimodular) subset of $\mathbb{U}$. Then $(p_{i}, \ket{\psi_{i}})_{i 
    \in \mathcal{I}}$ is an $\epsilon$-approximate homogeneous space state $t$-design if $\exists$ $\phi$ such that
    \begin{align}
        \left\|  \sum_{i \in \mathcal{I}} p_{i}  \Big(\ketbra{\psi_{i}}\Big)^{\otimes t} -\int_{\mathbb{U}/H}  \Big( \phi \bigg(\ketbra{x}\bigg) \Big)^{\otimes t} \dd_{\mathbb{U}/H}(x) \right\| < \epsilon.
    \end{align}
\end{definition}

It is important to appreciate that one can have two $\mathbb{U}/H$-equivalent state $t$-designs with different expectation values on a variety of operators. 

\subsection{Homogeneous space \texorpdfstring{$t$}{t}-designs}

 Now we proceed to define a homogeneous space unitary $t$-design directly on $\mathbb{U}/H$ similar to unitary $t$-designs on $\mathbb{U}$. Recall that these are subsets of $\mathbb{U}$ such that any polynomial function of degree $t$ (or less) in $\mathbb{U}$ averaged over the Haar measure can be replicated using the subset. When we move to homogeneous spaces, it is unclear as to how to define these polynomial functions at all. We begin by noting the following useful theorem, often called the \textit{quotient manifold theorem}.

\begin{theorem}
    Let $G$ be a Lie group and $M$ be a smooth manifold with a group action $\cdot: G \times M \rightarrow M$ such that $\cdot$ is free, smooth and proper. Then $M/G$ is a manifold, and it has a unique differentiable structure such that $\pi: M \rightarrow M/G$ is a submersion.
\end{theorem}

The details of this theorem are not too important for our purposes, other than the fact that it means there is a coordinate structure on $\mathbb{U}/H$ that resembles $\mathbb{R}^{d}$ for some $d$. A priori it is not clear how many coordinate charts are required to cover $\mathbb{U}/H$, but since the problems that we deal with are rooted in physics, we expect to largely avoid pathological cases. 

Before we discuss our definition of the homogeneous space $t$-design, let us discuss its utility. Its main purpose is to replace integrals over random ensembles of $\mathbb{U}$ with an expectation value over a finite number of summands. If the homogeneous space has only a finite number of equivalence classes, then in our case, this takes care of itself, but we also have the case where there are infinitely many such classes. The integrands are generally polynomials in the entries of the unitary matrices in the standard case. Let us use a similar idea to define the $t$-design in the case of the homogeneous space. Let us hence define the Homogeneous space $t$-design in the following way:

\begin{definition}
    Let $X \leq \mathbb{U}/H$, and $\Phi = \{\phi_{i}, i \in \mathcal{I}\}$ be an atlas for $\mathbb{U}/H$. Then $X$ forms a homogeneous space $t$-design if

    \begin{equation}
        \frac{1}{|X|} \sum_{x \in X} f(\Phi(x)) = \int_{\mathbb{U}/H} f(\Phi(x)) \mathrm{d}_{\mathbb{U}/H}(x)
    \end{equation}
    for all $f(x)$ with $f$ polynomial in $x$ with degree up to $t$.
\end{definition}

Notice that at the level of the unitaries, we are restricted to functions that take the same value for each of the members of a particular equivalence class under the symmetry group $H$ --- which in general need not be (and are likely not) polynomials in the space of unitary matrices. This is what we expect: if we consider any function that respects the symmetry of the group, then such a function is constant on the members of the equivalence class as well. Polynomial approximations to such functions in $\mathbb{U}/H$ form one example of a possible integrand.

While the definition looks rather restricted, we must be realistic --- we have lost a lot of structure in moving from $\mathbb{U}(d)$ to $\mathbb{U}/H$. In applications it is often useful in this framework to consider functions on $\mathbb{U}/H$ that are minima or maxima over the equivalence class elements of functions on $\mathbb{U}$ -- that is all $f$ such that $f:x \mapsto f(x) = \max g(U)$ for $U \in xH$ (or equivalently the minimum). It is also possible to define a generalized $t$-design when considering these functions as 

\begin{definition}
    Let $X \leq \mathbb{U}/H$, and $x = \{\phi_{i}: i \in \mathcal{I}\}$ be an atlas for $\mathbb{U}/H$. Then $X$ forms an extrema function homogeneous space $t$-design if

    \begin{equation}
        \frac{1}{|X|} \sum_{x \in X} f(x) = \int_{\mathbb{U}/H} f(x) \mathrm{d}_{\mathbb{U}/H}(x)
    \end{equation}
for all $f(x) = g(U)$ for $U \in xH$ with $g$ polynomial in the entries of $U$ with degree up to $t$.
\end{definition}

\section{Pseudorandom states and unitaries}
In this section, we adapt the definitions of pseudorandom unitaries (PRU) and pseudorandom states (PRS) from  Ref.~\cite{ji2018pseudorandom} %
for homogeneous spaces.
Consider a Hilbert space $\mathcal{H}$ and a key space $\mathcal{K}$, both dependent on a security parameter $\kappa$. Let $\mu_H$ be the Haar measure on $\mathbb{U\left(\mathcal{H}\right)}/H$ for a closed subgroup $H$, and let $\tau$ be an arbitrary (possibly probabilistic) map that assigns a representative in $\mathbb{U}$ for any equivalence class $x \in \mathbb{U}/H$.

\begin{definition}[Homogeneous space pseudorandom unitary operators (HPRUs)]
A family of unitary operators $\{U_k \in \mathbb{U\left(\mathcal{H}\right)}/H\}_{k \in \mathcal{K}}$ is homogeneous space pseudorandom if the following two conditions hold:

\begin{enumerate}
\item \textit{Efficient computation}: There exists an efficient quantum algorithm $Q$ such that for all $k\in \mathcal{K}$ and any pure state $\ket{\psi} \in \mathcal{H}$, $Q(k, \ket{\psi}) = U_k \ket{\psi}$.
\item \textit{Pseudorandomness}: $U_k$ with a random key $k$ is computationally indistinguishable from a homogeneous space Haar random unitary operator. More precisely, for any efficient quantum algorithm $A$ that makes at most polynomially many queries to the oracle,
\[
\left| \Pr_{k \gets K}[A^{U_k}(1^\kappa) = 1] - \Pr_{U \gets \nu}[A^U(1^\kappa) = 1] \right| = \operatorname{negl}(\kappa).
\]

where sampling according to $\nu$ refers to sampling equivalence classes in $\mu_{H}$ and then using the map $\tau$ to assign a representative unitary matrix. Note that each $\tau$ gives rise to a different $\nu$, and so in particular $\nu = \nu(\tau)$. For an ensemble to be an HPRU, we only require that there exists a single $\tau$ such that the sampling according to $\mu$ for that map satisfies the required relation.

\end{enumerate}
\end{definition}
Here $\operatorname{negl}(\kappa)$ denotes a negligible function, which can be any function that decays faster than an inverse polynomial.

This definition states that a family of unitary operators is considered pseudorandom if it is both efficiently computable and indistinguishable from Haar random unitary operators for an observer with a bounded computational power. Now let us proceed to define pseudorandom quantum states in a similar vein. Let us retain the definition of $\tau$ as an arbitrary (possibly probabilistic) map that assigns a representative to an equivalence class of unitaries, and $\nu$ as a sampling strategy that combines the Haar measure with $\tau$.

\begin{definition}[Homegeneous space pseudorandom quantum states (HPRSs)]
Let $\kappa$ be the security parameter. Consider a Hilbert space $\mathcal{H}$ and a key space $\mathcal{K}$, both dependent on $\kappa$. A keyed family of quantum states ${ \ket{\phi_k} \in \mathcal{S}(\mathcal{H}) }_{k \in \mathcal{K}}$ is defined as homogeneous space pseudorandom if it satisfies the following conditions:

\begin{enumerate}
\item \textit{Efficient generation}: There exists a polynomial-time quantum algorithm $G$ capable of generating the state $\ket{\phi_k}$ when given the input $k$. In other words, for every $k \in \mathcal{K}$, $G(k) = \ket{\phi_k}$.
\item \textit{Pseudorandomness}: When given the same random $k \in \mathcal{K}$, any polynomially bounded number of copies of $\ket{\phi_k}$ are computationally indistinguishable from the same number of copies of a Haar random state. More specifically, for any efficient quantum algorithm $A$ and any $m \in \operatorname{poly}(\kappa)$,
\[\left| \Pr_{k \gets K} [A(\ket{\phi_k}^{\otimes m}) = 1] - \Pr_{U \gets \nu} [A(U\ket{0}^{\otimes m}) = 1] \right| = \operatorname{negl}(\kappa).\]
\end{enumerate}
\end{definition}

In this definition, a keyed family of quantum states is considered pseudorandom if it can be generated efficiently and appears statistically indistinguishable from homogeneous space Haar random states to an observer with limited computational resources.

\begin{lemma}
    Consider two homogeneous spaces $\mathbb{U\left(\mathcal{H}\right)}/H_1$ and $\mathbb{U\left(\mathcal{H}\right)}/H_2$, such that $H_1 \leq H_2$. Every HPRU corresponding to $\mathbb{U\left(\mathcal{H}\right)}/H_1$ is an HPRU corresponding to $\mathbb{U\left(\mathcal{H}\right)}/H_2$. Similarly, every HPRS corresponding to $\mathbb{U\left(\mathcal{H}\right)}/H_1$ is an HPRS corresponding to $\mathbb{U\left(\mathcal{H}\right)}/H_2$. 
\end{lemma}
\begin{proof}
    Since $H_1 \leq H_2$, we have $\mathbb{U\left(\mathcal{H}\right)}/H_2 \leq \mathbb{U\left(\mathcal{H}\right)}/H_1$ and thus follows the above claim. 
\end{proof}

\begin{remark}
    Every PRS is an HPRS and similarly every PRU is an HPRU. This is because every subgroup $H$ contains identity as one of its elements.  
\end{remark}

\section{Application: The expressibility of parameterized  quantum circuits}

In this section, we will finally use the apparatus of homogeneous spaces to a useful NISQ task: quantum machine learning. The learning task we are interested in is as follows: we are given some input data describing properties of a system that are invariant under some symmetry transformation of the states, and want to try to learn the states which have this value of the property. For instance, if the property for which we are trying to learn the states is the entanglement (say we are given several measures of entanglement as our input data), then the symmetry group in question could consist of local operations on each subsystem (local unitary operations), as well as swapping around the subsystems. In this case, we do not mind which of the symmetry-related states the circuit outputs --- that is, as long as the state generates an output which has a representative in every equivalence class in the homogeneous space that is generated by this symmetry, it forms a suitable ansatz for our purposes.

Let us define a metric for the circuit ansatz in this symmetry-equivalent machine learning problem. The way we are going to do this is to measure how close the output distribution of states from this circuit ansatz is to every single equivalence class in the homogeneous space, with the measure of closeness defined ahead as a minimum distance to the elements in the class. Given a circuit ansatz $\mathcal{A}$, define the expressibility as (it can be shown equivalently) either

\begin{align}
    E^{\mathbb{U}/H}_{\mathcal{A}} &= \lim_{N \rightarrow \infty} \frac{1}{N} \sum_{i=1}^{N} \mathrm{min}_{{\theta}}  \bigg( \mathcal{D}_{\mathbb{U}/H} \bigg(\ket{i}_{\mathbb{U}/H},  \ket{\psi^{(i)}_{\mathcal{A}} (\theta)} \bigg) \bigg)
    \label{defk}
\end{align}
or 
\begin{align}   E^{\mathbb{U}/H}_{\mathcal{A}} &= \int_{{\theta}} \int_{i \in V^{(\mathbb{U}/H)}({\theta})} \mathcal{D}_{\mathbb{U}/H} \bigg(\ket{i}_{\mathbb{U}/H},  \ket{\psi_{\mathcal{A}} ({\theta})} \bigg) \nonumber\\
&\qquad\qquad \mathrm{d} \mu_{\mathbb{U}/H} (i) \mathrm{d}\theta.
\label{jfExpr}
\end{align}
Here $\ket{i}_{\mathbb{U}/H}$ is an arbitrary equivalence class of states generated by choosing $U \in \mathbb{U}/H$ according to the Haar measure and computing $\{\tilde{U}\ket{0} | \tilde{U} \in \pi^{-1}(U)\}$,  $\ket{\psi_{\mathcal{A}} ({\theta})}$ is the output ket of $\mathcal{A}$ for parameter choice ${\theta}$, and $V^{\mathbb{U}/H}({\theta})$ is the Voronoi cell produced using the distance $\mathcal{D}_{\mathbb{U}/H}$ around the state $\ket{\psi_{\mathcal{A}}({\theta})}$.  Note importantly that a \textit{smaller} value of $E$ defines a more expressible circuit.

The proof of their equivalence follows from the strong law of large numbers since 

\begin{align}
\
   E_{\mathcal{A}} &\xrightarrow{\mathrm{a.s.}} \mathbb{E}\left[\mathrm{min}_{\theta}  \bigg( \mathcal{D}_{G/H} \bigg(\ket{\psi^{(i)}_{\mathrm{Haar}}}_{G/H},  \ket{\psi^{(i)}_{\mathcal{A}} (\theta)} \bigg)\bigg)\right] \nonumber\\
   \implies
   \nonumber\\
   E_{\mathcal{A}} &= \int_{i \in G/H} \mathrm{min}_{\theta} \mathcal{D}_{G/H} \bigg(\ket{i}_{G/H},  \ket{\psi_{\mathcal{A}} (\theta)} \bigg) \mathrm{d} \mu_{G/H} (i),
\end{align}
which reduces to \Cref{jfExpr} as one partitions the volume into Voronoi cells, noting that the minimum value is achieved for a particular $\theta$ and Haar random state only if the sampled state is in the Voronoi cell of $\theta$.

Here the function $\mathcal{D}_{\mathbb{U}/H}$ is not a usual norm, since one of its arguments is sampled from $\mathbb{U}/H$, and the other (the state prepared by the ansatz) is in general in $\mathbb{U}$. We hence need a method to measure how far the equivalence class of states under a certain quotient map (say $\pi$) is from a particular state. For this we choose the most natural approach with 

\begin{equation}
        \mathcal{D}_{\mathbb{U}/H} \bigg(\ket{i}_{\mathbb{U}/H},  \ket{j} \bigg) = \min_{k} \: \mathcal{D} \bigg( \ket{k}, \ket{j} \bigg) \: k \in \pi^{-1}(\ket{i}_{\mathbb{U}/H}),
        \label{distance}
\end{equation}
where $\mathcal{D}$ is some distance function on $\mathbb{U}$. Note that in all of this, $H$ must be an explicitly defined subgroup of $\mathbb{U}$ --- for instance, we consider $\{I, X\}$ and $\{I,Y\}$ inequivalent, although they are both isomorphic to $\mathbb{Z}_{2}$. Let us now check that this definition gives us some of the results that we expect. We begin with the simple observation

\begin{theorem}
    Let $\mathcal{A}$ be an ansatz, and let $\alpha$ be a set of parametrized gates such that for all $a \in \alpha$, $a({0}) = \mathrm{id}$. Then $E^{\mathbb{U}/H}_{\mathcal{A}} \geq E^{\mathbb{U}/H}_{\mathcal{A} \cup \alpha}$.
\end{theorem}

\begin{proof}
    Let $S_{\mathcal{A}} = \{\ket{\psi_{\mathcal{A}}(\theta)}| \theta  \in \Theta_{\mathcal{A}}\}$ and $S_{\mathcal{A} \cup \alpha} = \{\ket{\psi_{\mathcal{A} \cup \alpha}(\theta)} | \theta  \in \Theta_{\mathcal{A} \cup \alpha}\}$, where $\Theta$ represents the range of the parameters in the circuit. But since setting $\theta_{i} = 0$ for all $\theta_{i} \in \alpha$ reproduces exactly $S_{\mathcal{A}}$, we have $S_{\mathcal{A}} \leq S_{\mathcal{A} \cup \alpha}$. Hence the Voronoi cells in \Cref{jfExpr} are strictly more granular, and we have the result.
\end{proof}

This is exactly as we would expect of our measure of expressibility --- adding more parametrized gates should not worsen this quantity, or equivalently increase $E^{\mathbb{U}/H}_{\mathcal{A}}$. We now discuss the property that if a circuit is highly expressible in the larger space $\mathbb{U}$, it should also be highly expressible in the more restricted space $\mathbb{U}/H$, proving it in two steps starting with the following theorem:

\begin{theorem}
    Let $H \leq \mathbb{U}$ with $\pi$ the canonical map $\pi: \mathbb{U} \rightarrow \mathbb{U}/H$. Let $K \leq \mathbb{U}$ be such that $\pi$ can be factorized as $\mathbb{U} \xrightarrow{\pi_{1}} \mathbb{U}/K \xrightarrow{\pi_{2}} \mathbb{U}/H$. Then $E_{\mathscr{A}}^{\mathbb{U}/H} \leq E_{\mathscr{A}}^{\mathbb{U}/K}$.
\end{theorem}

\begin{proof}
    Consider $\mathcal{D}_{G/H} \bigg(\ket{i}_{G/H},  \ket{j} \bigg)$ and let $S = \pi_{2}^{-1}(\ket{i}_{G/H}) \leq G/K$. Then we have $\mathcal{D}_{G/H} \bigg(\ket{i}_{G/H},  \ket{j} \bigg) \leq \mathcal{D}_{G/K} \bigg(\ket{k}_{G/K},  \ket{j} \bigg) \: \forall k \in S$ as $\pi_{1}^{-1}(\ket{k}_{G/K}) \leq \pi^{-1}(\ket{i}_{G/H})$. Hence $\min_{\theta} \mathcal{D}_{G/H} \bigg(\ket{i}_{G/H},  \ket{\psi(\theta)} \bigg)  \leq \bigg(\ket{k}_{G/K},  \ket{j} \bigg)$ so that $\ket{i}_{G/H}$ can never be found in a Voronoi cell such that its distance to the state $\ket{\psi_{\mathcal{A}}(\theta)}$ is worse than that for any state in $S$. But the measure assigned to $\ket{i}_{G/H}$ is the same as that assigned to $S$ by construction and \Cref{bigsubgrouphaarres}, and so we are done.
\end{proof}

\begin{corollary}
\label{OrderCorr}
    Let $H \leq \mathbb{U}$. Then $E_{\mathcal{A}}^{\mathbb{U}/H} \leq E_{\mathcal{A}}^{\mathbb{U}}$.
\end{corollary}

\begin{proof}
    This is a direct application of the previous theorem with $\pi_{1} = \mathrm{id}(G)$.
\end{proof}

The converse clearly is not true -- it is much easier to generate states close to \textit{one} certain representative of a coset than it is to generate states that are close to \textit{all} of them. In layman's terms, one can think that the symmetry in the system makes it easier to achieve randomness! We do however expect that appending a symmetry-creating unitary to an ansatz that is highly expressible in $\mathbb{U}/H$ should lead to high expressibility in $\mathbb{U}/H$. This turns out to be provably true if we pick the $\mathcal{D}$ in \Cref{distance} to be related to their overlap, or equivalently the Frobenius norm of their density matrices, so that we can use the unitarity property.

\begin{theorem}
    Let $\mathcal{A}$ be an ansatz and $U({\phi_{i}}), i \in I$ be a parametrized gate such that $\forall h \in H, \exists {\phi_{i}} \equiv {\phi_{h}}$ such that $U({\phi_{h}}) = h$. Further let $\mathcal{D} \bigg( \ket{k}, \ket{j} \bigg) = 1 - |\braket{k}{j}|^2$. Then $E^{\mathbb{U}}_{\mathcal{A} \cup U} \leq E^{\mathbb{U}/H}_{\mathcal{A}}$.  
\end{theorem}

\begin{proof}
    Define $\pi: G \rightarrow G/H$, and $S$ be the set of all $\ket{\psi} \in \pi^{-1}(\ket{i}_{G/H})$. We let $\ket{\psi_{\min}} \in S$ be such that if we take $\theta_{0}$ to be the angle that minimizes the main expression, then $\mathcal{D}_{G/H} \bigg(\ket{i}_{G/H},  \ket{\psi_{\mathcal{A}} (\theta_{0})} \bigg) = 1 - |\braket{\psi_{\min}}{\psi_{\mathcal{A}}(\theta_{0})}|^2$. But $\ket{\psi} \in S \implies \ket{\psi_{\min}} = h\ket{\psi}$ for some $h \in H$, and we have $\braket{\psi_{\min}}{\psi_{\mathcal{A}}(\theta_{0})}  = \bra{\psi_{\min}}h^{\dag}h\ket{\psi_{\mathcal{A}}(\theta_{0})} = \bra{\psi}h\ket{\psi_{\mathcal{A}}(\theta_{0})} = \braket{\psi}{\psi_{\mathcal{A} \cup U}(\theta_{0}, \phi_{h})}$. Now since $\ket{\psi}$ was arbitrary, we have that the distance $\mathcal{D}\bigg(\ket{\psi}, \ket{\psi_{\mathcal{A}}(\theta_{0}, \phi_{h})} \bigg) = \mathrm{min}_{\theta} \ \mathcal{D}_{G/H} \bigg(\ket{i}_{G/H},  \ket{\psi_{\mathcal{A}} (\theta)} \bigg)$ for all $\ket{\psi} \in S$, and hence $\mathrm{min}_{\theta} \mathcal{D}\bigg(\ket{\psi}, \ket{\psi_{\mathcal{A} \cup U} (\theta)} \bigg) \leq \mathrm{min}_{\theta} \ \mathcal{D}_{G/H} \bigg(\ket{i}_{G/H},  \ket{\psi_{\mathcal{A}} (\theta)} \bigg)$. But since by \Cref{bigsubgrouphaarres} $S \leq G$ and $\ket{i}_{G/H}$ are assigned the same volume in the integration, we are done.
\end{proof}

\begin{theorem}
   Let $\mathcal{A}$ be an ansatz, $H \leq \mathbb{U}$, $\mathcal{D} \bigg( \ket{k}, \ket{j} \bigg) = 1- |\braket{k}{j}|^2$. Further, let $d$ =  $\mathrm{dim}(\mathbb{U})$. Then $0 \leq E^{\mathbb{U}/H}_{\mathcal{A}} \leq 1-(\frac{1}{d})^2$.
\end{theorem}

\begin{proof}
    First notice that $\mathcal{D}_{G/H}$ is non-negative as the inner product on $G \geq 0$, so the integral must also be $0$ --- the integral of a non-negative function must also be non-negative. Thus we have the first part of the inequality. For the second part, we note that $E^{G/H}_{\mathcal{A}} \leq E^{G}_{\mathcal{A}}$ by Corollary \ref{OrderCorr}, and so computing $E^{G}_{\mathcal{A}}$ will provide an upper bound. This will be a circuit that prepares exactly one state, and by symmetry (or more technically the invariance of the Haar integral) it should not matter which, so we can just choose the Unitary operator to be the identity and hence the state to be $\ket{0}$. Hence what we require is to compute  $1 - |\int_{U \in G}  \bra{0}U\ket{0} \mathrm{d} \mu_{G} (U)|^2 = 1 - |\bra{0} \int_{U \in G} U \mathrm{d}\mu_{G}(U) \ket{0}|^2 = 1 - |\bra{0}{\frac{\mathbb{I}}{d}}\ket{0}|^2 = 1- \frac{1}{d^{2}}$. 
\end{proof}

\textit{Practical Remarks}. In general, there are two paths we can take for numerical simulation --- either we use the operational definition in \Cref{defk}, or compute the Voronoi cells and numerically evaluate the integral in \Cref{jfExpr}. The latter of these turns out to be rather complicated since we need to compute the Voronoi cells in a general homogeneous space with the state fidelity used as the distance.

Let us proceed instead with the operational method. Let $\mathcal{A}$ be the ($d$-qubit) ansatz we are considering, which takes parameters ${\theta} \in \Theta$. We begin by choosing $\ket{\psi} \sim \mathrm{Haar}(\mathbb{U})$, and choose $\{ {\theta}_{i} | \ i \in [N] \} \sim \mathrm{Uniform}(\Theta)$. This sampling complexity turns out to be the largest theoretical bottleneck, as the number of parameters scales exponentially as $2^{d}$. However, for small enough $d$ this remains computable. Then we can compute the set $\{ \ket{\psi_{\mathcal{A}}({\theta_{i}})}| \ i \in [N] \}$ efficiently if we have access to a quantum computer. Now we compute the $H-$equivalent representatives of $\ket{\psi}$, which we term $H_{\psi}$. If $H$ is discrete, this may be done simply by computing $\ket{\psi _h} = h\ket{\psi} \forall h \in H$; otherwise, we are required to produce some $\epsilon$-net or other finite tesselation for $H$. %
We then proceed to evaluate the quantity $\min_{h,i}\left(1 - |\braket{\psi _h}{\psi_{\mathcal{A}}({\theta_{i}})}|^2\right)$, which can be obtained efficiently on a quantum computer.
Alternatively, we can use a minimization routine to directly find this minimum distance. 
The average of this distance over many different $\ket{\psi}$ gives a good approximate to $E_{\mathcal{A}}^{\mathbb{U}/H}$, the expressibility of $\mathcal{A}$ in the subspace $\mathbb{U}/H$.

Notice that even though we are sampling as per the Haar measure on $\mathbb{U}$ rather than $\mathbb{U}/H$, we will get the correct result anyway, because of \Cref{bigsubgrouphaarres}. That is, we need $\mu_{\mathbb{U}/H}(1_{S}) = \mu_{\mathbb{U}}(1_{\pi^{-1}(S)})$ with $\pi$ as usual the canonical map $\mathbb{U} \rightarrow \mathbb{U}/H$ --- so as long as we make the correspondence $S \rightarrow \pi^{-1}(S)$, the sampling probability will be correct, and if we sample a certain point $x \in \mathbb{U}$, we take that to be equivalent to sampling the point $\pi(x) \in \mathbb{U}/H$ (with the usual understanding of sampling points from a continuous distribution).

\subsection{Numerical Results for \texorpdfstring{$\mathbb{U}(4)$, $H = \{\mathbb{I}, SWAP \}$}{U4, H={I,SWAP}}}
In this section, we present the results of the expressibility calculations for two different ansatze. First note that it is important to also include the encoding layer in the setup --- if not, we are restricting ourselves only to a subset of the total part of the Hilbert space the ansatz covers. We assume a simple angle encoding, wherein $\mathrm{RX}$ gates are used to encode the data in the two-qubit circuit.

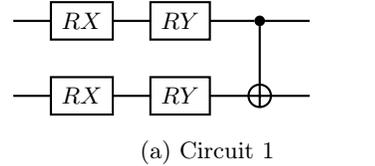
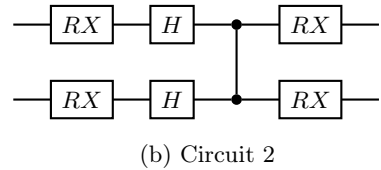
\begin{figure}[H]
    \centering
    \begin{subfigure}{0.3\textwidth}
            \begin{quantikz}
& \gate{RX} & \gate{RY} & \ctrl{1} & \qw \\
& \gate{RX} & \gate{RY} & \targ{} & \qw
\end{quantikz}
\caption{Circuit 1}
    \end{subfigure}
    \begin{subfigure}{0.3\textwidth}
    \begin{quantikz}
\\
& \gate{RX} & \gate{H} & \ctrl{1} & \gate{RX} & \qw\\
& \gate{RX} & \gate{H} & \ctrl{} & \gate{RX} & \qw
\end{quantikz}
\caption{Circuit 2}
    \end{subfigure}
\caption{(a) and (b) show the two circuits that we considered for the evaluation. The first pair of RX gates at the front represents the encoding.}
\end{figure}

\begin{figure}
\begin{subfigure}{0.35\textwidth}
\begin{tabular}{|l|l|l|}
\hline
{\color[HTML]{333333} Circuit \#} & {\color[HTML]{333333} $E^{\mathbb{U}(4)/H}_{\mathcal{A}}$} & {\color[HTML]{333333} Expr from \cite{Sim_2019}} \\ \hline
Circuit 1               
& $0.070 \pm 0.001$ 
& $0.095$  $ \pm 0.002$                                                                                            \\ \hline
Circuit 2                         & $0.204 \pm < 10^{-3}$                                          & $0.216 \pm 0.012$                                                   \\ \hline
\end{tabular}
    \end{subfigure}
    \caption{The table shows the values of the expressibilities calculated using the homogeneous expressibility scheme as well as the expressibility in \cite{Sim_2019}.}.

    \label{numerics}
\end{figure}

\section{Discussion}
This study investigates the interplay between symmetry and randomness in quantum systems, employing tools from homogeneous spaces. The novel concepts introduced, such as homogeneous space t-designs and homogeneous space pseudorandom states and unitaries, hold potential for various applications in quantum information. Designs refer to sets of objects and their groups, satisfying specific balance properties and symmetries. These have been extensively studied across diverse domains for centuries, ranging from error-correcting codes and card games to agriculture. Homogeneous space t-designs explicitly capture the symmetries of the underlying space and may prove beneficial in reducing resource requirements for tasks that involve t-designs, such as randomized benchmarking.

Similarly, quantum pseudorandom states and unitaries play crucial roles in various fields such as cryptography~\cite{ananth2022cryptography,ananth2023pseudorandom,haug2023pseudorandom}, complexity theory~\cite{kretschmer2021quantum}, learning theory~\cite{haug2023pseudorandom}, and black-hole physics~\cite{bouland2019computational}. We anticipate that the application of homogeneous space pseudorandom states and unitaries may hold promise in these aforementioned tasks, offering potential insights in scenarios involving symmetry. By harnessing the symmetries inherent in homogeneous spaces, these novel concepts have the potential to enhance the efficiency and effectiveness of algorithms and protocols across diverse quantum applications.

\paragraph*{Acknowledgements}
This research is supported by the National Research Foundation, Singapore and A*STAR under its Quantum Engineering Programme (NRF2021-QEP2-02-P03) and A*STAR (\#21709). DEK acknowledges funding support from the Agency for Science, Technology and Research (A*STAR) Central Research Fund (CRF) Award.

\bibliography{biblio}

\appendix

\clearpage

\onecolumngrid

\section{Limitations of existing approaches for quantifying expressibility} \label{sec:limitations}

Let us follow the calculation in Sim et al.~\cite{Sim_2019}. Here, they define the expressibility of the circuit as the difference between the ensemble of states generated by the VQC and an exact $t$-design, so if we define the operator $A$ 

\begin{equation}
    \mathrm{A} = \int_{\Uh, \Haar} \phi(\ketbra{\psi}{\psi})^{\otimes t} \dd \psi - \int_{\epsilon} (\ketbra{\phi}{\phi})^{\otimes t} \dd \phi,
\end{equation}

where we have directly used the definition in the main text. Then the squared Hilbert-Schmidt norm $\|\mathrm{A}\|^2 = \mathrm{Tr}(A^{\dag} A)$ is thought of as a measure of the expressibility for various $t$. The first challenge we face is finding the $\phi$ that minimizes the Hilbert-Schmidt norm --- otherwise we can only find an upper bound for it. In some sense, with different values of $\phi$, we get that the distribution generated by the states is a an $\epsilon$ approximate design for various $\epsilon$, and since we want a metric for it, we should find the lowest such $\epsilon$. Let us avoid this problem by picking a particular reasonable choice of strategy for the $\phi$ as follows: for every equivalence class $x \in \mathbb{U}/H$, $\phi$ maps $x$ to representatives with the same probability as the ensemble from the VQC.  With that out of the way we can follow \cite{Sim_2019} and write  $\mathrm{A} = \alpha^{(t)} - \mu^{(t)}$ for the two terms and get

\begin{equation}
    \mathrm{Tr}(A^{\dag} A) = \mathrm{Tr}((\alpha^{(t)} - \mu^{(t)})^{\dag} (\alpha^{(t)} - \mu^{(t)})).
\end{equation}

Let us suppress the superscript $t$ notation for now and then add it back in the end, keeping in mind we actually have one equation for every possible $t$. We can simplify this as 

\begin{equation}
    \mathrm{Tr}(A^{\dag} A) = \mathrm{Tr}(\alpha^{\dag}\alpha) + \mathrm{Tr}(\mu^{\dag}\mu) - 2\mathrm{Tr}(\alpha^{\dag}\mu),
    \label{mainTr}
\end{equation}
where we have implicitly used the fact that $\alpha$ and $\mu$ are Hermitian~\footnote{These are sums of projection matrices, so we can use linearity of the hermitian conjugate and the hermiticity of a projector.}. Now we can use the linearity of the trace to get

\begin{equation}
    \mathrm{Tr}(A^{\dag} A) = \mathrm{Tr}(\alpha^{\dag}(\alpha-2\mu)) + \mathrm{Tr}(\mu^{\dag}\mu).
\end{equation}

Now the issue that we have (in contrast to the case in \cite{Sim_2019}) is that we cannot simply replace the Haar integral with a projector on some symmetric subspace. Indeed the restriction to $\Uh$ means that we might not be able to commute the Haar integral with every operator on the symmetric subspace of $\mathbb{C}^{2^{n}}$, and so cannot use Schur's Lemma to conclude that this must be proportional to the identity matrix on this subspace. As a result, we cannot put the full focus of the $\|A^{\dag}A\|$ on just the part involving the Fidelity distribution of the output of the circuit as in \cite{Sim_2019}. It is still possible of course to find the KL divergence between these two distributions --- that of the (post-processed) Haar measure and that of the circuit --- but we lose whatever theoretical basis we had for that comparison.

\end{document}